\documentclass[aps,pra,10pt,twocolumn,superscriptaddress,floatfix,nofootinbib,showpacs,longbibliography,tikz,border=6mm]{revtex4-2}

\usepackage{amsmath}
\usepackage{tikz}

\usepackage[utf8]{inputenc}
\usepackage[T1]{fontenc}     
\usepackage[british]{babel}  
\usepackage[sc,osf]{mathpazo}
\usepackage{times}
\usepackage[table]{xcolor}
\usepackage[scaled=0.86]{berasans}  
\usepackage[colorlinks=true, citecolor=purple, urlcolor=blue]{hyperref}
\usepackage{comment}
\makeatletter
\newcommand{\setword}[2]{%
  \phantomsection
  #1\def\@currentlabel{\unexpanded{#1}}\label{#2}%
}
\makeatother
\usepackage{comment}
\usepackage{graphicx} 
\usepackage[babel]{microtype}  
\usepackage{amsmath,amssymb,amsthm,bm,amsfonts,mathrsfs,bbm} 

\usepackage{xspace}  
\usepackage{pgf,tikz}
\usepackage{xcolor}
\usepackage{multirow}
\usepackage{array}
\usepackage{bigstrut}
\usepackage{braket}
\usepackage{color}
\usepackage{natbib}
\usepackage{multirow}
\usepackage{mathtools}
\usepackage{float}
\usepackage[caption = false]{subfig}
\usepackage{xcolor,colortbl}
\usepackage{color}
\newcommand{\Tr}{\operatorname{Tr}}

\newcommand{\be}{\begin{equation}}
\newcommand{\ee}{\end{equation}}
\newcommand{\ba}{\begin{eqnarray}}
\newcommand{\ea}{\end{eqnarray}}

\newtheorem{theorem}{Theorem}
\newtheorem{corollary}{Corollary}
\newtheorem{definition}{Definition}

\newtheorem{lemma}{Lemma}

\def\>{\rangle}
\def\<{\langle}

\usepackage{centernot}
\usepackage{subfig}
\usepackage{filecontents}

\providecommand{\ket}[1]{| #1{\rangle}}
\providecommand{\bra}[1]{\langle #1|}

\usepackage[table]{xcolor}

\begin{document}

\title{The Generalised Causality Principle}

\author{Sahil Gopalkrishna Naik}
\affiliation{Department of Computer Science and Operations Research, Université de Montréal, Montréal, Québec, H3T 1J4, Canada}

\begin{abstract}
The no-signalling principle lies at the heart of Bell-type experiments involving spacelike-separated parties. Although not all no-signalling correlations can be realized within quantum theory, the no-sinalling principle itself remains fundamental even for post-quantum theories, since its violation would imply superluminal signalling. In the spirit of device independent nature of the Bell inequalities, there has recently been growing interest in causal inequalities, which identify correlations admitting no causal explanation. Such correlations have been rigorously studied within the process-matrix formalism, which is built upon a constraint analogous to spacelike separation in Bell scenarios. We refer to this constraint as the single-interaction constraint, whereby each party interacts with the environment only once. Despite considerable effort, no fundamental guiding principle is currently known that characterizes the restrictions this constraint imposes on the correlations generated by Process Matrices. In this work, we address this gap by proposing a generalized causality principle for the bipartite case. We show that for certain cases, the generalised causality principle defines a strict subset of the set of all signalling correlations. This principle thus provides a fundamental limit on quantum and even post-quantum theories subject to the single-interaction constraint. By analogy with no-signalling, whose violation excludes certain causal structures among the parties, a violation of the generalised causality principle rules out a corresponding class of causal structures(causal structures with the single interaction constraint). More broadly, our techniques provide a general framework for deriving generalised causal constraints, applicable to a wide class of indefinite-causal-order scenarios.
\end{abstract}


\maketitle	
\section{Introduction}
The everyday notion of causality rests on an operational picture of spacetime: for any pair of events, either a signal can be sent from one to the other, or signalling is impossible in both directions. Within quantum theory, such causally ordered processes are described by quantum circuits and, more generally, by quantum higher-order maps~\cite{Chiribella2008,Chiribella2009}. However, several lines of thought suggest that, at a more fundamental level, reconciling the dynamical causal structure of general relativity with the probabilistic nature of quantum mechanics may require abandoning a definite underlying time or causal order~\cite{Hardy2007}. Indeed, indefinite causal order already arises within the framework of higher-order quantum theory itself, as exemplified by the quantum switch, in which a control system coherently determines the order in which two processes are applied~\cite{Chiribella2012,Chiribella2013,Arajo2014,Procopio2015}.

Over a decade ago, Oreshkov, Costa, and Brukner~\cite{Oreshkov2012} proposed a device-independent approach to testing such indefinite causal structures. In their framework, each party is situated in a closed, isolated laboratory and interacts with the external environment only once: a quantum system from the environment enters the laboratory, the party applies a quantum instrument with multiple classical outcomes, and the resulting evolved state is returned to the environment. Crucially, once the system has left the laboratory, no further information can re-enter it, so that no second interaction with the environment is possible. When every party operates under this \emph{single-interaction constraint}, the observed correlations are fully determined by the state of the environment, known in the literature as the \emph{process matrix}. Such process matrices can generate correlations that admit no causal explanation, as witnessed by the violation of causal inequalities~\cite{Oreshkov2012}. Since then, the structure and properties of process matrices have been studied extensively~\cite{Oreshkov2016,Baumeler2016,Kunjwal2023,Naik2025}.

In analogy with the Tsirelson bound in Bell--CHSH-type experiments, considerable effort has been devoted to determining the maximal quantum violations of causal inequalities within the process-matrix framework~\cite{Branciard2015}. Recently, Liu \emph{et al.}~\cite{Liu2025} showed that the exact quantum violation can be computed for a class of causal inequalities known as single-trigger inequalities, and that their approach also yields an upper bound on the violation of arbitrary causal inequalities. In a related direction, Bavaresco \emph{et al.}~\cite{Bavaresco2024} introduced process matrices within Boxworld-type generalised probabilistic theories (GPTs) and showed that, under suitable constraints, the set of correlations obtainable from the Boxworld GPT forms a strict subset of the set of all correlations. They further conjectured that these Boxworld correlations provide an outer approximation to the set of quantum correlations generated by quantum instruments and quantum processes.

Motivated by these results, we address the following fundamental question: \emph{just as the constraint of spacelike separation imposes the no-signalling condition on the correlations observable in the Bell scenario, does the single-interaction constraint impose an analogous, theory-independent restriction on the correlations observable within the process-matrix framework?} We answer this question in the affirmative for the bipartite case by introducing a generalisation of the no-signalling principle to the process-matrix scenario, which we call the \emph{generalised causality principle}. We show that, when both parties implement instruments with binary outcomes, this principle imposes no additional restriction on the set of observed correlations. In contrast, when both parties implement three instruments with three outcomes each, it imposes a nontrivial restriction. The set of correlations satisfying the generalised causality principle forms a convex polytope, which we call the \emph{Generalised Causality Polytope}, and which plays a role analogous to that of the no-signalling polytope in the Bell scenario. Accordingly, just as a violation of the no-signalling principle can be interpreted as evidence of superluminal influence, a violation of the generalised causality principle can be interpreted as evidence of a hypothetical signalling influence that does not satisfy the single interaction constraint.

To establish these results, we begin in Section~\ref{prelimanaries} by introducing the necessary preliminaries. In contrast to the approach of \cite{Bavaresco2024}, where boxworld objects are defined as probability distributions, we formulate arbitrary generalised probabilistic theories (GPTs) in the language of Hilbert-space operators. We show that, as in quantum theory, objects in GPTs can be represented by Hermitian operators subject to appropriate constraints. Building on this formalism, Section~\ref{results} presents the main results of this work. Finally, in Section~\ref{discussion}, we discuss the broader implications and applicability of our results.

\section{Preliminaries}\label{prelimanaries}

In this section, we introduce the framework of general probabilistic theories (GPTs) in the Hilbert space operator representation, which will be convenient for the analysis in this work.

\subsection{General probabilistic theories}

In the framework of general probabilistic theories a system $\mathcal{S}$ is specified by a pair
\begin{equation}
    \mathcal{S} = (\Omega_{\mathcal{S}}
, \mathcal{E}_{\mathcal{S}}),
\end{equation}
where $\Omega_{\mathcal{S}}$ denotes the normalised state space and $\mathcal{E}_{\mathcal{S}}$ denotes the effect space. Throughout this work we assume the no restriction hypothesis which states that any object (measurement effect, transformations) that consistently provides predictions must be an actual component in the theory.

To construct the underlying vector space structure of a GPT, we associate with the system a $d_{\mathcal{S}}$-dimensional Hilbert space $H_{\mathcal{S}}$. Let us denote the set of generalized Pauli operators acting on $H_{\mathcal{S}}$ as
\begin{align*}
\mathcal{P}_{\mathcal{S}} =&\left\{\sigma^{\mu}_{\mathcal{S}}\right\}
_{\mu=0}^{d_{\mathcal{S}}^2-1},~\mbox{with}~\Tr[\sigma^{\mu}_{\mathcal{S}}
\sigma^{\nu}_{\mathcal{S}}]=d_{\mathcal{S}}\delta_{\mu\nu},\\
&\left(\sigma^{\mu}_{\mathcal{S}}\right)^2=\mathbb{I}_{\mathcal{S}}~\mbox{and}~\sigma^{0}_{\mathcal{S}}=\mathbb{I}_{\mathcal{S}}    
\end{align*}

Where $\mathbb{I}_{\mathcal{S}}$ denotes the identity operator. For the system $\mathcal{S}$ we denote the set of Hermitian operators of $\mathcal{S}$ as
\begin{equation}
\operatorname{Herm}(\mathcal{S}):=\operatorname{span}_{\mathbb{R}}(\mathcal{P'}_{\mathcal{S}})~\mbox{where}~
\mathcal{P'}_{\mathcal{S}}=\left\{\sigma^{\mu}_{\mathcal{S}}\right\}
_{\mu=0}^{l_{\mathcal{S}}-1}
\end{equation}
Here we take $l_{\mathcal{S}}\leq d^2_{\mathcal{S}}$ as for general systems $d_{\mathcal{S}}^2$ Pauli's might not be necessary to denote the state space. For instance for the square bit theory it suffices to take $d_{\mathcal{S}}=2$ and $l_{\mathcal{S}}=3$ as the state space spans a $3$ dimensional real vector space. Formally $l_{\mathcal{S}}$ is known as the linear dimension of the system $\mathcal{S}$.

For a general GPT the normalised state space $\Omega_{\mathcal{S}}$ is then defined as a convex subset of $\operatorname{Herm}(\mathcal{S})$ as follows
\begin{equation}
    \Omega_{\mathcal{S}}:=\left\{\omega_{\mathcal{S}}\in \operatorname{Herm}(\mathcal{S})~|~\Tr[\omega_{\mathcal{S}}]=1\right\}
\end{equation}
Every state $\omega_{\mathcal{S}} \in \Omega_{\mathcal{S}}$ can then be written as,
\begin{equation}
\omega_{\mathcal{S}}=\frac{1}{d}\left[\mathbb{I}_{\mathcal{S}}+\sum_{j=1}^{l_{\mathcal{S}-1}}r_j\sigma^j_{\mathcal{S}}\right],\qquad r_j \in \mathbb{R}.
\end{equation}
Without loss of generality we assume that $\mathbb{I}/{d_{S}}$ can always be included in the interior of $\Omega_{\mathcal{S}}$ for any system $\mathcal{S}$. Under the no restriction hypothesis we then define a valid effect $e_{\mathcal{S}}\in\mathcal{E}_{\mathcal{S}}$ as an element of the dual vector space $\operatorname{Herm}(\mathcal{S})^{*}$ that maps states to probabilities as follows 
\begin{align}
    \mathcal{E}_{\mathcal{S}}=\left\{e_{\mathcal{S}}\in \operatorname{Herm}(\mathcal{S})^{*}~|~0\leq\Tr[e_{\mathcal{S}}\omega_{\mathcal{S}}]\leq1~\forall~\omega_{\mathcal{S}}\in\Omega_{\mathcal{S}}\right\}
\end{align}

Similar to ordinary Quantum theory the role of the unit effect is taken by the identity operator as $\Tr[\mathbb{I}_{\mathcal{S}}\omega_{\mathcal{S}}]=1~\forall~\omega_{\mathcal{S}}$. We may occasionally omit the subscript of an operator when its associated system is clear from the context.

\subsection{Transformations among GPTs}
Given a input system $\mathcal{S}_1$ and an output system $\mathcal{S}_2$ a probabilistic transformation $\mathcal{T}:\operatorname{Herm}(\mathcal{S}_1)\mapsto\operatorname{Herm}(\mathcal{S}_2)$ is a linear map satisfying
\begin{align}
\mathcal{T}(\omega)\in\operatorname{Cone}(\Omega_{\mathcal{S}_2})~\mbox{and}~\Tr[\mathcal{T}(\omega)]\leq1~\forall~\omega\in\Omega_{\mathcal{S}_1}~
\end{align}
Here the $\operatorname{Cone}$ function denotes the conical hull of the set appearing in its argument. A transformation is termed as a deterministic transformation if $\mathcal{T}(\Omega_{\mathcal{S}_1})\subseteq\Omega_{\mathcal{S}_2}$. The probabilistic and deterministic transformations play the role of trace-nonincreasing and trace preserving transformations in Quantum Theory respectively.

\emph{Choi–Jamiołkowski isomorphism}\cite{Choi1975}: Similar to quantum theory we define the notion of a maximally entangled operator on two systems $\mathcal{S}$ and $\mathcal{S'}$ which are isomorphic to each other as follows
\begin{align}
    \Phi_{SS'}=\frac{1}{d_\mathcal{S}}\left(\sum_{\mu=0}^{l_\mathcal{S}-1}\sigma^{\mu}_{\mathcal{S}}\otimes\sigma^{\mu}_{\mathcal{S}'}\right)
\end{align}
Denoting $\mathcal{I}$ as the identity transformation we define the Choi operator $T$ for transformation $\mathcal{T}$ as\footnote{For Quantum theory we use the usual definition of Choi Isomorphism by taking $\Phi=\sum_{ij}\ket{i}\bra{j}\otimes\ket{i}\bra{j}$ as the unnormalised maximally entangled state.}
\begin{align}
T_{\mathcal{S}_1\mathcal{S}_2}:=\mathcal{I}\otimes\mathcal{T}\left(\Phi_{\mathcal{S}_1\mathcal{S}'_1}\right)
\end{align}
A crucial point to note is that as the transformation $\mathcal{T}$ is a hermiticity preserving transformation the Choi operator $T$ will always be a hermitian operator. As the above equation defines an isomorphism it is possible to compute the action of $\mathcal{T}$ from $T$ as follows
\begin{align}
\mathcal{T}\left(\omega_{\mathcal{S}_1}\right)=\Tr_{\mathcal{S}_1}\left[\left(\omega_{\mathcal{S}_1}\otimes\mathbb{I}_{\mathcal{S}_2}\right)T_{\mathcal{S}_1\mathcal{S}_2}\right]
\end{align}
For deterministic transformations we have
\begin{align}
&\Tr[\mathcal{T}\left(\omega\right)]=1~\forall\omega\implies _{\mathcal{S}_2}T_{\mathcal{S}_1\mathcal{S}_2}=\frac{\mathbb{I}_{\mathcal{S}_1\mathcal{S}_2}}{d_{\mathcal{S}_2}}\\
&\mbox{Where}~_{X}T=\frac{\mathbb{I_X}}{d_X}\otimes\Tr_X[T].\nonumber
\end{align}
A generalised instrument $\mathbf{I}_{\mathcal{S}_1\mathcal{S}_2}$ is a collection of probabilistic transformations that sum up to a deterministic transformation.
\begin{align}
\mathbf{I}_{\mathcal{S}_1\mathcal{S}_2}=\left\{T_{\mathcal{S}_1\mathcal{S}_2}^k~\Big{|}~_{\mathcal{S}_2}\left(\sum_kT_{\mathcal{S}_1\mathcal{S}_2}^k\right)=\frac{\mathbb{I}_{\mathcal{S}_1\mathcal{S}_2}}{d_{\mathcal{S}_2}}\right\}
\end{align}
Given a state $\omega$ as an input to the instrument the index $k$ represents a classical outcome generated by the instrument.  $\mathcal{T}^k(\omega)$ denotes the transformed subnormalised state while the probability for obtaining the $k^{th}$ outcome is given by $\Tr[\mathcal{T}^k(\omega)]$.

\subsection{Bipartite Process Matrices in GPTs}

Given an agent Alice we define the system $A_I$ as a system entering Alice's Laboratory. Alice then performs any transformation of this system to a possibly different system $A_O$. The set of all probabilistic transformations for Alice is denoted by $\mathbf{T}_{\mathcal{A}}:=\left\{T_{A_IA_O}\right\}$. Similarly for Bob we define the set $\mathbf{T}_{\mathcal{B}}:=\left\{T_{B_IB_O}\right\}$. We also define the subsets $\mathbf{T}^{D}_{\mathcal{A}}\subseteq\mathbf{T}_{\mathcal{A}}$ and $\mathbf{T}^{D}_{\mathcal{B}}\subseteq\mathbf{T}_{\mathcal{B}}$ containing only deterministic transformations.
\begin{definition}
A bipartite hermitian process $W_{A_IA_OB_IB_O}$ is a four partite hermitian matrix that describes the state of the environment of Alice's and Bob's Laboratories. Such a process is a bilinear map that yields consistent probabilities for all transformations applied by Alice and Bob i.e.
\begin{subequations}\label{ConstraintsW}
  \begin{align}
&0\leq\Tr[W_{A_IA_OB_IB_O}\left(T_{A_IA_O}\otimes T_{B_IB_O}\right)]\leq1~\nonumber\\
&\forall~T_{A_IA_O}\in\mathbf{T}_{\mathcal{A}},\forall~T_{B_IB_O}\in\mathbf{T}_{\mathcal{B}}~\mbox{and}\label{positivityW}
\end{align}
\begin{align}
&\Tr[W_{A_IA_OB_IB_O}\left(T_{A_IA_O}\otimes T_{B_IB_O}\right)]=1~\nonumber\\
&\forall~T_{A_IA_O}\in\mathbf{T}^{D}_{\mathcal{A}},\forall~T_{B_IB_O}\in\mathbf{T}^{D}_{\mathcal{B}}.\label{normalisationW}
  \end{align}
\end{subequations}
\end{definition}
In the following we derive the constraints on $W$ imposed by Eq.(\ref{ConstraintsW}). It turns out that the positivity constraints in Eq.(\ref{positivityW}) are difficult to analyse in general and depends on the specific systems $A_I,A_O,B_I$ and $B_O$. Thus now we concentrate on the normalisation constraints Eq.(\ref{normalisationW}). Before moving on to the explicit analysis of the constraint we first introduce a useful lemma
\begin{lemma}\label{spanlemma}
Under the no restriction hypothesis the set $\mathbf{T}_{\mathcal{A}}$ spans the space of all Hermitian Operators associated with the composite system $A_I$ and $A_O$ i.e.
\begin{align*}
\operatorname{Span}_{\mathbb{R}}\mathbf{T}_{\mathcal{A}}=\operatorname{Span}_{\mathbb{R}}\left\{\sigma^{\mu}_{A_I}\otimes\sigma^{\nu}_{A_O}\right\}_{\mu=0,\nu=0}^{l_{A_I}-1,l_{A_O}-1}
\end{align*}
\end{lemma}
The proof is provided in the appendix \ref{Proof of Spanlemma}. The normalisation conditions in Eq.(\ref{normalisationW}) imply

\begin{subequations}\label{Processnormalisation}
\begin{align}
\Tr\left[W\right]&=d_{A_O}d_{B_O}\label{c1}\\  
_{B_IB_O}W&=_{A_OB_IB_O}W\label{c2}\\
_{A_IA_O}W&=_{B_OA_IA_O}W\label{c3}\\
W+_{A_OB_O}W&=_{A_O}W+_{B_O}W.\label{c4}
\end{align}
\end{subequations}
See Appendix \ref{Normalisation constraints for Hermitian Processes} for detailed derivation. It is interesting to note that the normalisation conditions on Hermitian processes coincide directly with the normalisation conditions of Quantum Process Matrices\cite{Oreshkov2012,Arajo2015}.

\emph{Causally separable processes} A Hermitian Process $W_{A_IA_OB_IB_O}$ is said to be causally separable if it can be written in the following form
\begin{align}
    W=pW^{A\not\prec B}+(1-p)W^{B\not\prec A}~~~p\in[0,1]
\end{align}
Where $W^{A\not\prec B}(W^{B\not\prec A})$ is a valid hermitian process where Alice(Bob) cannot signal to Bob(Alice). Such processes also satisfy

\begin{subequations}\label{cs0}
\begin{align} 
&W^{B\not\prec A}= _{B_O}W^{B\not\prec A},~
_{B_IB_O}W^{B\not\prec A}= _{A_OB_IB_O}W^{B\not\prec A},\label{cs0a}\\
&W^{A\not\prec B}= _{A_O}W^{A\not\prec B},~
_{A_IA_O}W^{A\not\prec B}= _{B_OA_IA_O}W^{A\not\prec B}.\label{cs0b}\\
&~~~~~~~~~~~~~~~~~W^{A\nprec\nsucc B}=_{A_OB_O}W^{A\nprec\nsucc B}
\end{align}
\end{subequations} 
If a process is no-signalling in both direction we denoted it by $W^{A\nprec\nsucc B}$. The above conditions are derived in detail in Appendix \ref{Constraints for Causal Hermitian Process}.

\subsection{Hermitian Correlations}

A correlation $\{p(ab|xy)\}$ denotes the probability of Alice and Bob obtaining the classical outcomes $a$ and $b$ given that they have implemented the instrument $x$ and $y$ respectively. It is assumed that every party interacts with the environment only once i.e. a system from the environment enters their lab after which they implement an instrument on the system, labelled by $x(y)$ to obtain a classical outcome $a(b)$ and then, they send the evolved state outside their lab back to the environment. The correlations achieved via this single interaction constraint are given by
\begin{align}\label{correlation}
    p(ab|xy)=\Tr[W_{A_IA_OB_IB_O}(T^{a|x}_{A_IA_O}\otimes T^{b|y}_{B_IB_O})]
\end{align}
where $T_{A_IA_O}^{a|x}\in\mathbf{T}_{\mathcal{A}}~\forall~a,x$, $T_{B_IB_O}^{b|y}\in\mathbf{T}_{\mathcal{B}}~\forall~b,y$ and $\sum_aT_{A_IA_O}^{a|x}:=T_{A_IA_O}^{~|x}\in\mathbf{T}^{D}_{\mathcal{A}}~\forall~x$,$\sum_bT_{B_IB_O}^{b|y}:=T_{B_IB_O}^{~|y}\in\mathbf{T}^{D}_{\mathcal{B}}~\forall~y$.
\begin{definition}
A correlation  $\{p(ab|xy)\}$ is called a hermitian correlation if $\exists$ systems $A_I,A_O,B_I$ and $B_O$, a valid hermitian process $W_{A_IA_OB_IB_O}$ and some instruments $\{T_{A_IA_O}^{a|x}\},\{T_{B_IB_O}^{b|y}\}$ such that Eq.(\ref{correlation}) is satisfied.  
\end{definition}

\begin{definition}
A correlation $\{p(ab|xy)\}$ is said to be causal if it can be expressed as
\begin{align}
p(ab|xy)=q p^{A\not \prec B}(ab|xy) + (1-q) p^{B\not \prec A}(ab|xy)\nonumber
\end{align}
Where $\{p^{A(B)\not \prec B(A)}(ab|xy)\}$ signifies that Alice(Bob) is not in the Causal past of Bob(Alice). More formally 
\begin{subequations}
\begin{align}
&\sum_{a}p^{A\not \prec B}(ab|xy)=\sum_{a}p^{A\not \prec B}(ab|x'y)~\forall~b,y,x,x'\label{Causalcorr1}\\
&\sum_{b}p^{B\not \prec A}(ab|xy)=\sum_{b}p^{B\not \prec A}(ab|xy')~\forall~a,x,y,y'\label{Causalcorr2}
\end{align}
\end{subequations}
We denote $p^{A\not \prec \not \succ B}(ab|xy)$ as a no-signalling correlation if both equations above are satisfied.
\end{definition}

Denoting the cardinalites of $a,b,x$ and $y$ as $O_A,O_B,I_A$ and $I_B$ respectively, we define the full correlation polytope as follows
\begin{definition}
The full correlation polytope $\mathbf{FC}(I_A,I_B,O_A,O_B)$ is the largest set of consistent probability distributions i.e.
\begin{align}
&\mathbf{FC}(I_A,I_B,O_A,O_B)=\nonumber\\
&\left\{\{p(ab|xy)\}~\Big{|}~p(ab|xy)\geq0,\sum_{ab}p(ab|xy)=1\right\}    
\end{align}
\end{definition}
\section{Results}\label{results}
We now proceed to the main contribution of this work. We first state the the definition of an average causal correlation.
\begin{definition}
A correlation $\{p(ab|xy)\}$ is called an average causal correlation if it can be expressed as 
\begin{align}
    p(ab|xy)=\frac{1}{2}p^{A\not \prec B}(ab|xy)+\frac{1}{2}p^{B\not \prec A}(ab|xy)
\end{align}
\end{definition}
It is important to note that due to the convex mixture an average causal correlation has restricted signalling power. Given the definition of an average causal correlation the generalised causality principle is defined as follows:

\begin{definition}\label{GCPdefnition}
A correlation $\{p(ab|xy)\}$ is said to satisfy generalised causality principle (GCP) if there exists a no-signalling correlation $\{p^{A \not \prec \not \succ B}(ab|xy)\}$ such that the average correlation 
\begin{align}
    \frac{1}{2}p(ab|xy)+\frac{1}{2}p^{A \not \prec \not \succ B}(ab|xy)
\end{align}
has an average causal explanation. Equivalently,
\begin{align}
    p(ab|xy)=&p^{A\not \prec B}(ab|xy)+p^{B\not \prec A}(ab|xy)\nonumber\\
    &-p^{A \not \prec \not \succ B}(ab|xy)~\forall~a,b,x,y\label{GCP}
\end{align}
\end{definition}
To appreciate the content of this principle, it is instructive to consider a correlation that violates it. For such a correlation, an equal mixture with any no-signalling correlation admits no average causal explanation. Equivalently, its signalling strength exceeds the critical limit permitted by the principle. We now prove the main result of our paper

\begin{theorem}\label{GCPTheorem}
Every hermitian correlation satisfies the generalised causality principle.
\end{theorem}
The proof is discussed in the Appendix \ref{Proof of GCPTheorem}.

\begin{corollary}
All causal correlations satisfy GCP.
\end{corollary}
\begin{proof}
If $\{p(ab|xy)\}$ is a correlation where Alice cannot signal to Bob we simply set $p^{A\not \prec B}(ab|xy):=p(ab|xy)$ and $p^{B\not \prec A}(ab|xy):=p^{A \not \prec \not \succ B}(ab|xy)$ where $p^{A \not \prec \not \succ B}(ab|xy)$ can be any arbitary no-signalling correlation. Similary we show the same for correlations where Bob cannot signal to Alice. For general Causal correlations it suffices to note that convex mixture of correlations satisfying GCP also satisfies GCP.
\end{proof}

\begin{definition}
Given $O_A,O_B,I_A$ and $I_B$ the convex set of all correlations that satisfy the GCP is called the Generalised Causality Polytope which we denote by 
\begin{align}
&\mathbf{GC}(I_A,I_B,O_A,O_B)=\nonumber\\
&\left\{\{p(ab|xy)\}~\Big{|}~\{p(ab|xy)\} ~\mbox{satifies GCP}~\right\}    
\end{align}
\end{definition}
For finite input and output alphabets, the generalised causality condition reduces to a finite set of linear constraints on $\{p(ab|xy)\}$, so the set of compatible correlations is a polytope. We begin by showing that, if all instruments have binary outcomes, generalised causality is satisfied by every correlation $\{p(ab|xy)\}$.

\begin{theorem}
$\mathbf{GC}(I_A,I_B,2,2)=\mathbf{FC}(I_A,I_B,2,2).$   
\end{theorem}
\begin{proof}
It suffices to show that all extreme points of $\mathbf{FC}(I_A,I_B,2,2)$ can indeed be expressed in the form of Eq.(\ref{GCP}). The extreme points of $\mathbf{FC}(I_A,I_B,2,2)$ are given as follows
\begin{align}
p(ab|xy)=\delta_{af(x,y)}\delta_{bg(x,y)}
\end{align}
Where, $f$ and $g$ are binary functions on $x$ and $y$. Defining
\begin{align}
p^{A\not\prec B}(ab|xy)&:=\frac{\delta_{af(x,y)}}{2}\nonumber\\
p^{B\not\prec A}(ab|xy)&:=\frac{\delta_{bg(x,y)}}{2}\nonumber\\
p^{A\not\prec\not\succ B}(ab|xy)&:=\frac{\delta_{af(x,y)}+\delta_{bg(x,y)}}{2}-\delta_{af(x,y)}\delta_{bg(x,y)}\nonumber
\end{align}
It can be checked that the above correlations satisfy the required no-sinalling conditions. This completes the proof.
\end{proof}
Although these sets coincide in this special case, for a larger number of outputs the generalised causality principle imposes nontrivial restrictions on the correlations that can be achieved using Hermitian processes, as demonstrated by the following theorem.

\begin{theorem}
$\mathbf{GC}(3,3,3,3)\subsetneq\mathbf{FC}(3,3,3,3).$   
\end{theorem}
\begin{proof}
We show this by explicitly demonstrating that the Guess your Neighbour's input type  correlation $\{p(ab|xy)=\delta_{ay}\delta_{bx}\}$ does not belong to $\mathbf{GC}(3,3,3,3)$. We start by assuming that this correlation satisfies GCP. Thus for $a=y$ and $b=x$ we have
\begin{align}
&p^{A \not \prec \not \succ B}(a=yb=x|xy)=p^{A\not \prec B}(a=yb=x|xy)\nonumber\\
&~~~~~~~~~~~~~~+p^{B\not \prec A}(a=yb=x|xy)-1   \end{align}
Since $p^{A \not \prec \not \succ B}(a=yb=x|xy)\geq0$ we have
\begin{align}
\sum_{xy}p^{A\not \prec B}(a=yb=x|xy)+p^{B\not \prec A}(a=yb=x|xy) \geq 9 \label{9ineq}  
\end{align}
But for the first term we have
\begin{align}
\sum_{xy}p^{A\not \prec B}(a=yb=x|xy)&\leq \sum_{xya}p^{A\not \prec B}(ab=x|xy)\nonumber\\
=\sum_{xy}~p^{A\not \prec B}&(b=x|xy)\nonumber\\
=\sum_{xy}~p^{A\not \prec B}&(b=x|y)=\sum_y 1 =3
\end{align}
Similarly for the second term we have 
\begin{align}
\sum_{xy}p^{B\not \prec A}(a=yb=x|xy)&\leq3
\end{align}
This is a contradiction as the left hand side in Eq.(\ref{9ineq}) can atmost be $6$.
\end{proof}
The above theorem thus successfully demonstrates that all signalling correlations are not attainable by hermitian processes in any GPT model.

\section{Discussion}\label{discussion}
The generalised causality principle imposes an explicit constraint on the correlations that can be observed in generalised probabilistic theories (GPTs) within the process-matrix scenario. The core argument underlying Theorem~\ref{GCPTheorem} rests on the normalisation constraint on process matrices. Our results show that, whereas positivity constraints are specific to quantum theory, normalisation constraints are universal across GPTs. This observation opens a natural direction for future work: the same approach can be extended to multipartite processes, and even to multi-round process matrices \cite{Hoffreumon2021}, to investigate the implications of normalisation constraints on such higher-order processes across all GPTs.

Our work also raises several questions of fundamental and mathematical interest. The first is whether every correlation satisfying the generalised causality principle can be reproduced by some GPT model. In the Bell scenario, an analogous statement is known to hold: multipartite Boxworld states suffice to reproduce any no-signalling correlation \cite{Plvala2025}. Establishing a corresponding result in the process-matrix setting would clarify the scope of GPTs as a framework for describing correlations without a definite causal order, and would be a significant mathematical result in its own right.

A second line of enquiry concerns the geometry of the Generalised Causality Polytope. Characterising the extremal points of $\mathbf{GC}(3,3,3,3)$ is a natural goal, but a direct computation is prohibitively expensive: for three inputs and three outputs, the no-signalling polytope already has a very high dimension and a large number of facets, and only partial results on its facet Bell inequalities and extremal points are currently available \cite{Cope2019}. It would nevertheless be of great interest to identify at least some extremal points of $\mathbf{GC}(3,3,3,3)$ that play a role analogous to that of the Popescu--Rohrlich box \cite{Popescu1994} in the Bell scenario. Such correlations could serve as canonical examples of maximal violations of causal inequalities, and could offer new insight into the structure of correlations compatible with the single interaction constraint.

\noindent{\bf Acknowledgement}:SGN acknowledges support from the Courtois Scientific Vanguard Fund program of the Faculty of Arts and Sciences and the Department of Computer Science and Operational Research (DIRO) of the Université de Montréal. SGN thanks Dr.Kuntal SenGupta for helpful discussions related to this work.
\bibliography{HP}

\appendix
\onecolumngrid
\section{Proof of Lemma \ref{spanlemma}}\label{Proof of Spanlemma}
\begin{proof}
From no restriction hypothesis we have transformation $T_{A_IA_O}:=\frac{\mathbb{I}_{A_IA_O}}{d_{A_O}}\in\mathbf{T}_{\mathcal{A}}$ as this transformation maps every state to $\frac{\mathbb{I}_{A_O}}{d_{A_O}}$. Now we prove the lemma via a contradiction and assume that there exists a orthogonal basis $H^{\alpha}_{A_IA_O}$ orthogonal to $\operatorname{Span}_{\mathbb{R}}\mathbf{T}_{\mathcal{A}}$ such that 
\begin{align*}
\operatorname{Span}_{\mathbb{R}}\left(\mathbf{T}_{\mathcal{A}}\cup_{\alpha}H^{\alpha}_{A_IA_O}\right)=\operatorname{Span}_{\mathbb{R}}\left\{\sigma^{\mu}_{A_I}\otimes\sigma^{\nu}_{A_O}\right\}_{\mu=0,\nu=0}^{l_{A_I}-1,l_{A_O}-1}    
\end{align*}
Note that the transformations $H^{\alpha}_{A_IA_O}$ though not valid can always be considered to be trace non-increasing transformations as we can always multiply a positive scalar to $H^{\alpha}_{A_IA_O}$ to make it trace non increasing while maintaining the orthogonality. We now define the hermitian operators 
\begin{align*}
H^{\alpha,\lambda}_{A_IA_O}=\lambda H^{\alpha}_{A_IA_O}+(1-\lambda)\frac{\mathbb{I}_{A_IA_O}}{d_{A_O}}~\mbox{with}~\lambda\in(0,1]   
\end{align*}
Note that since $\frac{\mathbb{I}_{A_IA_O}}{d_{A_O}}\in\mathbf{T}_{\mathcal{A}}$ we also have
\begin{align*}
\operatorname{Span}_{\mathbb{R}}\left(\mathbf{T}_{\mathcal{A}}\cup_{\alpha}H^{\alpha,\lambda}_{A_IA_O}\right)=\operatorname{Span}_{\mathbb{R}}\left\{\sigma^{\mu}_{A_I}\otimes\sigma^{\nu}_{A_O}\right\}_{\mu=0,\nu=0}^{l_{A_I}-1,l_{A_O}-1}    
\end{align*}
Note that there exists a sufficiently small $\lambda$ such that  we have 
\begin{align*}
\Tr_{A_I}\left[\left(\omega_{A_I}\otimes\mathbb{I}_{A_O}\right)H^{\alpha,\lambda}_{A_IA_O}\right]=\lambda\mathcal{H}^{\alpha}(\omega_{A_I})+(1-\lambda)\frac{\mathbb{I}_{A_O}}{d_{A_O}}\in\operatorname{Cone}(\Omega_{A_O})~~\forall\omega_{A_I}\in\Omega_{A_I}
\end{align*}
Here $\mathcal{H}^{\alpha}$ denote the transformations corrsponding to the Choi operators $H^{\alpha}_{A_IA_O}$. Moreover the transformations $H^{\alpha,\lambda}_{A_IA_O}$ are trace non-increasing. But according to the definition of these operators we have $H^{\alpha,\lambda}_{A_IA_O}\notin\mathbf{T}_{\mathcal{A}}$. Which is a contradiction as according to the no restriction hypothesis such operators must be in $\mathbf{T}_{\mathcal{A}}$.
\end{proof}

\section{Normalisation constraints for Hermitian Processes}\label{Normalisation constraints for Hermitian Processes}

Since the set $\mathbf{T}_{\mathcal{A}}$ and $\mathbf{T}_{\mathcal{B}}$ span the space of  hermitian operators we can impose the condition Eq.(\ref{normalisationW}) on arbitrary Hermitian operators that are trace preserving or deterministic. Note that for arbitrary hermitian operators $X_{A_IA_O}$ and $Y_{B_IB_O}$ the hermitian operators $X'_{A_IA_O}=X_{A_IA_O}-_{A_O}X_{A_IA_O}+\frac{\mathbb{I}_{A_IA_O}}{d_{A_O}}$ and $Y'_{B_IB_O}=Y_{B_IB_O}-_{B_O}Y_{B_IB_O}+\frac{\mathbb{I}_{B_IB_O}}{d_{B_O}}$ are always trace preserving. Thus we impose
\begin{align*}
\Tr\left[W\left(X'_{A_IA_O}\otimes Y'_{B_IB_O}\right)\right]=1~\forall~X_{A_IA_O},Y_{B_IB_O}    
\end{align*}
Substituting $X=Y=0$ we get
\begin{align*}
    \Tr\left[W\right]=d_{A_O}d_{B_O}
\end{align*}
If $Y=0$ and $X=0$, respectively, we have
\begin{align*}
    \Tr\left[W\left((X-_{A_O}X)\otimes\mathbb{I}\right)\right]=0~\forall~X
\end{align*}
\begin{align*}
    \Tr\left[W\left(\mathbb{I}\otimes(Y-_{A_O}Y)\right)\right]=0~\forall~Y
\end{align*}
And finally, we have
\begin{align*}
    \Tr\left[W\left((X-_{A_O}X)\otimes(Y-_{A_O}Y)\right)\right]=0~\forall~X,Y
\end{align*}
As the transformations of the form $_{\mathcal{S}}(\cdot)$ are self adjoint the above equalities yield

\begin{subequations}
\begin{align}
\Tr[\left(\Tr_{B_IB_O}W\right)X]=\Tr[_{A_O}\left(\Tr_{B_IB_O}W\right)X]~\forall~ X&\implies_{B_IB_O}W=_{A_OB_IB_O}W\nonumber\\
\Tr[\left(\Tr_{A_IA_O}W\right)Y]=\Tr[_{B_O}\left(\Tr_{A_IA_O}W\right)Y]~\forall~ Y&\implies_{A_IA_O}W=_{B_OA_IA_O}W\nonumber\\
\Tr[W\left(X\otimes Y\right)]=\Tr[\left(_{A_O}W+_{B_O}W-_{A_OB_O}W\right)\left(X\otimes Y\right)]~\forall~X,Y&\implies W=_{A_O}W+_{B_O}W-_{A_OB_O}W.\nonumber
\end{align}
\end{subequations}

\section{Constraints for Causal Hermitian Process}\label{Constraints for Causal Hermitian Process}

Firstly, we note that a single-party hermitian process satisfies
\begin{align}
    0\leq\Tr[W_{A_IA_O}T_{A_IA_O}]\leq1~ \forall~ T_{A_IA_O}\in \mathbf{T}_{\mathcal{A}}
\end{align}
along with the normalisation constraints
\begin{subequations}
\begin{align}
\Tr[W]=d_{A_O}\\
W=_{A_O}W
\end{align}
\end{subequations}

We now consider the bipartite hermitian process where Bob cannot signal to Alice, denoted as $W^{B\not\prec A}$.  
Every operation $T_{A_IA_O}$ performed by Alice induces a reduced sub process for Bob, which may depend functionally on $T_{A_IA_O}$, defined as
\begin{align}\label{r}
W_{B_IB_O}(T_{A_IA_O}) = \Tr_{A_IA_O}\!\left[W^{B\not\prec A}_{A_IA_OB_IB_O}(T_{A_IA_O}\otimes\mathbb{I}_{B_IB_O})\right].
\end{align}
Although $W_{B_IB_O}(T_{A_IA_O})$ is not, in general, a bonafide single-party hermitian process (since $T_{A_IA_O}$ need not be trace-preserving), it is nevertheless a linear functional of the form
\begin{subequations}
\begin{align}
    0\leq\Tr[W_{B_IB_O}(T_{A_IA_O})T_{B_IB_O}]\leq p(T_{A_IA_O})~ \forall~ T_{B_IB_O}\in \mathbf{T}_{\mathcal{B}}\\
    \Tr[W_{B_IB_O}(T_{A_IA_O})T_{B_IB_O}]=p(T_{A_IA_O})~ \forall~ T_{B_IB_O}\in \mathbf{T}^{D}_{\mathcal{B}}
\end{align}
\end{subequations}
The second equality in the above equation reflects the fact that Bob’s overall outcome probability coincides with that of Alice’s operation $T_{A_IA_O}$,as Bob cannot signal to Alice. Since, $p(T_{A_IA_O})$ just serves as a proportionality factor for normalization for the subprocess $W_{B_IB_O}(T_{A_IA_O})$, we have
\begin{subequations}
\begin{align}
\Tr\left[W_{B_IB_O}(T_{A_IA_O})\right] &= p(T_{A_IA_O})d_{B_O}, \quad\forall~T_{A_IA_O}\in\mathbf{T}_{\mathcal{A}}\label{n1}\\    
W_{B_IB_O}(T_{A_IA_O}) &=_{B_O}\!\left(W_{B_IB_O}(T_{A_IA_O})\right).\quad\forall~T_{A_IA_O}\mathbf{T}_{\mathcal{A}} \label{n2}
\end{align}
\end{subequations} 
Substituting Eq.~(\ref{r}) into the expressions above yields
\begin{subequations}
\begin{align}
&\Tr_{A_IA_OB_IB_O}\!\left[W^{B\not\prec A}_{A_IA_OB_IB_O}\!\left(T_{A_IA_O}\otimes\frac{1}{d_{B_O}}\mathbb{I}_{B_IB_O}\right)\right] = p(T_{A_IA_O}), \label{m1}\\
&\Tr_{A_IA_O}\!\left(W^{B\not\prec A}_{A_IA_OB_IB_O}(T_{A_IA_O}\otimes\mathbb{I}_{B_IB_O})\right) = 
\Tr_{A_IA_O}\!\left(({}_{B_O}W^{B\not\prec A}_{A_IA_OB_IB_O})(T_{A_IA_O}\otimes\mathbb{I}_{B_IB_O})\right). \label{m2}
\end{align}
\end{subequations}
Equation~(\ref{m1}) is trivially satisfied, since $\frac{1}{d_{B_O}}\mathbb{I}_{B_IB_O}$ corresponds to a trace-preserving operation.  
On the other hand, the requirement that Eq.~(\ref{m2}) holds for all $T_{A_IA_O}\in\mathbf{T}_{\mathcal{A}}$ implies
\begin{align*}
W^{B\not\prec A} = {}_{B_O}W^{B\not\prec A},
\end{align*}
which yields the first constraint in Eq.(\ref{cs0a}).  

Next, we derive the second constraint. Since Bob cannot signal to Alice, the reduced single-party process for Alice is independent of  the deterministic transformation applied by Bob and thus must be a constant process $W^{\mathrm{const}}_{A_IA_O}$, i.e.,
\begin{align*}
\Tr_{B_IB_O}\!\left[W^{B\not\prec A}_{A_IA_OB_IB_O}(\mathbb{I}_{A_IA_O}\otimes T_{B_IB_O})\right]
= W^{\mathrm{const}}_{A_IA_O}.\forall~ T_{B_IB_O}\in \mathbf{T}^{D}_{\mathcal{B}}
\end{align*}
Since this holds for all deterministic transformations of Bob, the reduced normalized process $W^{\mathrm{const}}_{A_IA_O}$ can be obtained by substituting $T_{B_IB_O}=\frac{\mathbb{I}_{B_IB_O}}{d_{B_O}}$, giving
\begin{align*}
W^{\mathrm{const}}_{A_IA_O}\otimes\mathbb{I}_{B_IB_O}
= \frac{1}{d_{B_O}}\Tr_{B_IB_O}\!\left[W^{B\not\prec A}_{A_IA_OB_IB_O}\right]\otimes\mathbb{I}_{B_IB_O}
= d_{B_I}\left({}_{B_IB_O}W^{B\not\prec A}_{A_IA_OB_IB_O}\right).
\end{align*}
Since $W^{\mathrm{const}}_{A_IA_O}$ is a valid single-party hermitian process, we have
\begin{align*}
W^{\mathrm{const}}_{A_IA_O}=_{A_O}W^{\mathrm{const}}_{A_IA_O}\implies_{B_IB_O}W^{B\not\prec A} = {}_{A_OB_IB_O}W^{B\not\prec A}.
\end{align*}
The other normalization constraint $\Tr[W^{\mathrm{const}}_{A_IA_O}]=d_{A_O}$ follows directly from the normalization of the bipartite process $W^{B\not\prec A}$.  
This completes the derivation of Eq.(\ref{cs0b}).  
The analogous relations for $W^{A\not\prec B}$ follow straightforwardly.  
When Alice and Bob are spacelike separated, the combined conditions for $W^{B\not\prec A}$ and $W^{A\not\prec B}$ imply
\begin{align*}
W^{A\nprec\nsucc B} = {}_{A_OB_O}W^{A\nprec\nsucc B}.
\end{align*}

\section{Proof of Theorem \ref{GCPTheorem}}\label{Proof of GCPTheorem}

\begin{proof}
We start by noting that for any quantum or post quantum theory the process $W$ satisfies $W=_{A_O}W+_{B_O}W-_{A_OB_O}W$ (From Eq.(\ref{c4}).It is easy to see that $_{A_O}W$ satisfies the normalisation constraints in Eq.(\ref{Processnormalisation}). We now show that $_{A_O}W$ also yields consistent probabilities and hence is a valid Hermitian Process. Firstly we can always write the action of $T_{A_IA_O}\in\mathbf{T}_{\mathcal{A}}$ on $\omega_{A_I}\in\Omega_{A_I}$ as
\begin{align}
\Tr_{A_I}\left[\left(\omega_{A_I}\otimes\mathbb{I}_{A_O}\right)T_{A_IA_O}\right]=p\omega_{A_O}
\end{align}
Where $0\leq p\leq1$ and $\omega_{A_O}\in\Omega_{A_O}$. Thus we also have
\begin{align}
\Tr_{A_I}\left[\left(\omega_{A_I}\otimes\mathbb{I}_{A_O}\right)_{A_O}T_{A_IA_O}\right]=p_{A_O}\omega_{A_O}=p\frac{\mathbb{I}_{A_O}}{d_{A_O}}
\end{align}
Thus $T_{A_IA_O}\in\mathbf{T}_{\mathcal{A}}\implies _{A_O}T_{A_IA_O}\in\mathbf{T}_{\mathcal{A}}$.Since $W$ is a valid hermitian process, we have
\begin{align}
    \Tr[_{A_O}W\left(T_{A_IA_O}\otimes T_{B_IB_O}\right)]= \Tr[W\left(_{A_O}T_{A_IA_O}\otimes T_{B_IB_O}\right)]\geq 0 ~\forall ~T_{A_IA_O}\in \mathbf{T}_{\mathcal{A}},T_{B_IB_O}\in\mathbf{T}_{\mathcal{B}}
\end{align}
Thus $_{A_O}W$ is a valid hermitian process. It is also easy to verify that $_{A_O}W$ also satisfies Eq.(\ref{cs0b}) and can only lead to correlations where Alice cannot signal to Bob. Similary $_{B_O}W$ leads to correlations where Bob cannot signal to Alice. Finally $_{A_OB_O}W$ always results in no-signalling correlations. Thus any hermitian correlation $\{p(ab|xy)\}$ can be written as
\begin{align}
p(ab|xy)&=\Tr\left[W\left(T^{a|x}_{A_IA_O}\otimes T^{b|y}_{B_IB_O}\right)\right]=\Tr\left[\left(_{A_O}W+_{B_O}W-_{A_OB_O}W\right)\left(T^{a|x}_{A_IA_O}\otimes T^{b|y}_{B_IB_O}\right)\right]\nonumber\\
&=p^{A\not \prec B}(ab|xy)+p^{B\not \prec A}(ab|xy)-p^{A \not \prec \not \succ B}(ab|xy)~\forall~a,b,x,y
\end{align}
\end{proof}

\end{document}